\documentclass[copyright,creativecommons]{eptcs}
\providecommand{\event}{E. Formenti (Ed.): AUTOMATA \& JAC 2012}
\usepackage{amssymb,amsmath,graphicx}
\usepackage[utf8x]{inputenc}

\title{Local Rules for Computable Planar Tilings}
\author{Thomas Fernique
\institute{LIPN\\CNRS\& Université Paris 13, Sorbonne Paris Cité}
\email{thomas.fernique@lipn.univ-paris13.fr}
\and
Mathieu Sablik
\institute{LATP\\Université Aix-Marseille}
\email{mathieu.sablik@latp.univ-mrs.fr}
}
\def\titlerunning{Local rules for computable planar tilings}
\def\authorrunning{Th. Fernique, M. Sablik}

\newtheorem{theorem}{Theorem}
\newtheorem{proposition}{Proposition}
\newenvironment{proof}{\noindent \emph{Proof. }}{\hfill \hbox{\rlap{$\sqcap$}$\sqcup$}\\}

\begin{document}
\maketitle

\begin{abstract}
Aperiodic tilings are non-periodic tilings characterized by local constraints.
They play a key role in the proof of the undecidability of the domino problem (1964) and naturally model quasicrystals (discovered in 1982).
A central question is to cha\-rac\-te\-ri\-ze, among a class of non-periodic tilings, the aperiodic ones.
In this paper, we answer this question for the well-studied class of non-periodic tilings obtained by digitizing irrational vector spaces.
Namely, we prove that such tilings are aperiodic if and only if the digitized vector spaces are computable.
\end{abstract}

\section{Introduction}

A {\em tiling} is a covering of a given space by interior-disjoint compacts called {\em tiles}.
The shape of tiles yields constraints  on the way tiles can locally be arranged -- one speaks about {\em local rules} (think, {\em e.g.}, of the bumps and dents of a jigsaw puzzle).
Additional local rules can also be set by arbitrarily specifying how tiles can be neighboor (tiles can also be available in different colors in order to allow them playing different roles).\\

Tilings have been studied in computer sciences in the early 60's by the logician Hao Wang, who set the so-called {\em domino problem} (\cite{wang}): can one decide whether a given finite set of tiles can form a tiling of the plane (each tile can be used several times)?
His student Robert Berger proved the undecidability of this problem (\cite{berger}).
The two key ingredients of the proof are, first, the simulation of Turing computations by tilings of the plane and, second, the existence of {\em aperiodic} tile sets, that are finite tile sets which do tile the plane but only in a non-periodic fashion (Berger explicitly described the first ever such tile set).\\

The interest in aperiodic tilings ({\em i.e.}, tilings by aperiodic tile sets) spreaded beyond computer sciences two decades later, when new non-periodic crystals soon called {\em quasicrystals} were incidentally discovered by the chemist Dan Shechtman (\cite{shechtman}).
The connection with aperiodic tilings was indeed quickly done, with tiles and local rules res\-pec\-ti\-ve\-ly modelling atom clusters and finite range energetic interactions.
The issue that now concerns theoretical physicists is the classification of all the possible quasicrystalline structure, in the spirit of the Bravais-Fedorov classification of crystalline structures.\\

A promising approach is the one opened by Leonid Levitov in \cite{levitov}.
He considered non-periodic planar tilings, that are digitizations of irrational vector spaces, and searched algebraic conditions on vector space parameters for the existence of local rules.
This approach led to numerous results (\cite{burkov,le1,le2,le3,le4,socolar1,socolar2}), but no complete characterization of aperiodic planar tilings has yet been obtained.\\

The aim of this paper is to move a step forward in the above approach by enriching geometric methods with calculability, in the spirit of the first works on aperiodic tile sets.
Our main result (Theorem~\ref{th:main}, below) is that a planar tiling admits local rules if and only if it is a digitization of a vector space whose parameters are computable.\\

Our result thus provides a complete characterization -- at least if we do not care about parameters as the size of tile sets or the precision of digitizations.
Of course, these parameters are important {\em w.r.t.} the quasicrystal modelization.
In particular, the tile sets obtained via the Levitov approach are much smaller than the huge ones that we get by simulating computations of Turing machines (the planar tilings obtained in the former case are however restricted to algebraic parameters, versus computable parameters in the latter case).
The quest for sharper classifications of aperiodic planar tilings thus remains open, with our result nevertheless limiting the horizon of any such classification.\\

The rest of the paper is organized as follows.
In Section~\ref{sec:formalization}, we introduce the formalism which allows us to state us our main result, Theorem~\ref{th:main}.
Section~\ref{sec:reciproque} shows that one cannot expect to go beyond computability with local rules.
The following sections are devoted to prove that local rules allow to indeed reach the computability barrier.
Specifically, Section~\ref{sec:quasisturmian} introduces quasi-Sturmian words, which are particular non-periodic words whose letters are indexed by $\mathbb{Z}$.
Local rules cannot characterize such words, but Section~\ref{sec:subshift} shows that they can characterize the set of two-dimensional words (letters indexed by $\mathbb{Z}^2$) whose lines are quasi-Sturmian words.
The key ingredient is a result simultaneously obtained by \cite{AS} and \cite{DRS}.
We then show in Section~\ref{sec:3_vers_2} how to transfer this result onto planar tilings which are digitization of planes in $\mathbb{R}^3$, and we finally extend this to any planar tiling in Sec.~\ref{sec:n_vers_d}.

\section{Formalization}
\label{sec:formalization}
\paragraph{$\mathbf{n\to d}$ \textbf{tilings}\\}
Let $\vec{v}_1,\ldots,\vec{v}_n$ be pairwise non-colinear vectors of $\mathbb{R}^d$, $n>d>0$.
A {\em $n\to d$ tile} is, up to translation, a parallelotope generated by $d$ of the $\vec{v}_i$'s, {\em i.e.}, the linear combinations with coefficient in $[0,1]$ of $d$ of the $\vec{v}_i$'s.
Then, a {\em $n\to d$ tiling} is a face-to-face tiling of $\mathbb{R}^d$ by $n\to d$ tiles, {\em i.e.}, a covering of $\mathbb{R}^d$ by $n\to d$ tiles which can intersect only on full faces of dimension less than $d$.

\paragraph{\textbf{Lift}\\}
Let $\vec{e}_1,\ldots,\vec{e}_n$ be the canonical basis of $\mathbb{R}^n$.
Given a $n\to d$ tiling, we first arbitrarily map one of its vertex to $\vec{0}\in\mathbb{R}^n$, then we map each tile generated by $\vec{v}_{i_1},\ldots,\vec{v}_{i_d}$ onto the $d$-dimensional face of a unit hypercube of $\mathbb{Z}^n$ generated by $\vec{e}_{i_1},\ldots,\vec{e}_{i_d}$, with two tiles adjacent along an edge $\vec{v}_i$ being mapped onto two faces adjacent along an edge $\vec{e}_i$.
This defines, up to the choice of the initial vertex, the {\em lift} of the tiling. This is a digital $d$-dimensional manifold in $\mathbb{R}^n$, whence $d$ and $n-d$ are respectively called the dimension and the codimension of the tiling.

\paragraph{\textbf{Computable planar tilings}\\}
A $n\to d$ tiling is said to be {\em planar} if there are a $d$-dimensional vector subspace $V\subset \mathbb{R}^n$ and a positive integer $w$ such that this tiling can be lifted into the slice $V+[0,w)^n$.
The space $V$ is called the {\em slope} of the tiling and the smallest suitable $w$ its {\em width} (both are uniquely defined).
A planar tiling is said to be computable if its slope is computable, {\em i.e.}, admits a basis of vectors with computable coordinates that is to say they can be  computed to within any desired precision by a Turing machine.

\paragraph{\textbf{Local rules}\\}
A planar $n\to d$ tiling of slope $V$ is said to admit {\em local rules} if, when $n\to d$ tiles are available in colors chosen among a given finite set, the way two colored tiles can intersect can be thoroughly specified so that the allowed colored tilings form a non-empty set of colored planar $n\to d$ tilings of slope $V$ and uniformly bounded width $w$.
In other words, the slope $V$ is characterized by finitely many rules governing the way tiles locally match, with colors allowing each tile to play different roles.
In \cite{levitov}, local rules are said to be {\em strong} if $w=1$, {\em weak} otherwise.\\

\noindent We are now in a position to state our main result:
\begin{theorem}\label{th:main}
A planar tiling admits local rules if and only if it is computable.
\end{theorem}

\section{The Computability barrier}
\label{sec:reciproque}

Consider $\mathbb{R}^n$ endowed with the norm $||\vec{v}||_{\infty}$ which gives the maximum of the absolute value of the coordinates of $\vec{v}$. Let $\mathbf{S}=\{\vec{x}\in\mathbb{R}^n : ||x||_{\infty}=1\}$, for two $d$-dimensional vector spaces $V$ and $W$ of $\mathbb{R}^n$, we can define the distance 
$$\widetilde{d}(V,W)=\max\left\{\sup_{\vec{v}\in\mathbf{S}\cap V}\inf_{\vec{w} \in W}||\vec{v}-\vec{w}||_{\infty}\ ;\ \sup_{\vec{w}\in\mathbf{S}\cap W}\inf_{\vec{v} \in V}||\vec{v}-\vec{w}||_{\infty}\right\}.$$
Since $V$ and $W$ have the same dimension, then the maximum in the expression above is always attained by both expressions simultaneously. The set of $d$-dimensional vector spaces is compact by this distance. Moreover if we know the computable basis of $V$ and $W$ then $\widetilde{d}(V,W)$ is also computable. With this formalism one obtains an equivalent definition of computable vector spaces: $V$ is computable if there exists a Turing machine such that on the input $n\in\mathbb{N}$ it gives a rational basis of a vector space $W_n$ such that $\widetilde{d}(V,W_n)\leq\frac{1}{n}$.\\

\noindent We here show the easiest part of Theorem~\ref{th:main}:

\begin{proposition}
If a planar tiling admits local rules, then it is computable.
\end{proposition}

\begin{proof}
Consider a planar $n\to d$ tiling of slope $V$ and width $w$ which admits local rules and take $P_r$ the set of all the diameter $r$ patterns centred on $\vec{0}$ of colored tilings allowed by these local rules (this takes exponential but finite time in $r$). Let $\mathcal{X}_r$ be the set of $d$-dimensional vector spaces which admit a basis given by $d$ vectors  associated at a border vertex in the lift of a pattern of $P_r$. The set $\mathcal{X}_r$ is finite, moreover there exists $W\in \mathcal{X}_r$ such that $\widetilde{d}(W,V)\leq \frac{w}{r}$. Since for sufficiently large $r$ all vector spaces of $\mathcal{X}_r$ are near of $V$ (if not by compacity one obtains one other slope for the $n\to d$ tiling), to obtain an approximation of $V$ with an error bound $\epsilon$, we take the first $r>\frac{2w}{\epsilon}$ such that for all $W_1,W_2\in\mathcal{X}_r$ one has $\widetilde{d}(W_1,W_2)<\frac{\epsilon}{2}$. In this case, for all $W\in\mathcal{X}_r$ one has $\widetilde{d}(V,W)<\epsilon$ so all basis of norm one of $W$ are an approximation of a basis of V. Thus $V$ is a computable vector space. 
\end{proof}

\section{Quasi-Sturmian words}
\label{sec:quasisturmian}

Consider the set $\{0,1\}^{\mathbb{Z}}$ of bi-infinite words over the alphabet $\{0,1\}$ endowed with the metric $d$ defined, for any $u$ and $v$, by
$$
d(u,v):=\sup_{p\leq q}\left||u(p)u(p+1)\ldots u(q)|_0-|v(p)v(p+1)\ldots v(q)|_0\right|,
$$
where $|w|_0$ denote the number of occurences of the letter $0$ in the finite word $w$.
In other terms, the distance between two words is the maximum {\em balance} between their finite factors which begin and start at the same positions.\\

\noindent 
Define the Sturmian word $s_{\rho,\alpha}\in\{0,1\}^{\mathbb{Z}}$ of slope $\alpha\in[0,1]$ and intercept $\rho$ by
$$
s_{\rho,\alpha}(n)=0
~\Leftrightarrow~
(\rho+n\alpha)\mod 1\in[0,1-\alpha).
$$
Sturmian words have been extensively studied (see, {\em e.g.}, \cite{lothaire} for a detailed account).
In particular, they can be seen as $2\to 1$ planar tilings of width $1$ (the digitized vector space is here a line of slope $\alpha$, and the tiles are letters $0$ and $1$).
Classic properties of Sturmian words easily yield:

\begin{proposition}\label{prop:parallel_sturmian}
Sturmian words with equal slopes are at distance at most one.
\end{proposition}

\begin{proof}
Two sturmian words $u$ and $v$ with equal slopes are known to have the same finite factors.
Any two factors of respectively $u$ and $v$ which begin and start at the same positions are thus also factors of $u$ only - at different position but with the same number of letters.
This yields the bound
$$
d(u,v)\leq\sup_{p,q,r}\left||u(p)u(p+1)\ldots u(p+r)|_0-|u(q)u(q+1)\ldots u(q+r)|_0\right|.
$$
This bound is known to be at most one for Sturmian words (and only them).
\end{proof}

The words at distance at most one from a Sturmian word of slope $\alpha$ are however not all Sturmian.
We call them {\em quasi-Sturmian} (of slope $\alpha$).
They can be seen as $2\to 1$ planar tilings of width $2$.
The easy following proposition will be useful to link Sturmian and quasi-Sturmian words:

\begin{proposition}\label{prop:alternated_replacements}
Two words in $\{0,1\}^{\mathbb{Z}}$ are at distance at most one if and only if each can be obtained from the other by performing letter replacements $0\to 1$ or $1\to 0$, without two consecutive replacements of the same type.
\end{proposition}

\begin{proof}
Let $u$ and $v$ in $\{0,1\}^{\mathbb{Z}}$ at distance at most one.
Performing on $u$ replacements at each position $i$ where $u(i)\neq v(i)$ yields $v$.
If two consecutive replacements, say at position $p$ and $q$, have the same type, then the balance between $u(p)\ldots u(q)$ and $v(p)\ldots v(q)$ is two, hence $d(u,v)\geq 2$.
The type of replacements thus necessarily alternates.\\
Conversely, assume that $v\in\{0,1\}^{\mathbb{Z}}$ is obtained from $u\in\{0,1\}^{\mathbb{Z}}$ by per\-for\-ming replacements whose type alternates.
Given $p\leq q$, consider the number of replacements between positions $p$ and $q$: the balance between $u(p)\ldots u(q)$ and $v(p)\ldots v(q)$ is $0$ if this number is even, $1$ otherwise, hence $d(u,v)\leq 1$.
\end{proof}

Since the replacements to transform $u$ in $v$ alternate, their sequence can be encoded by $w\in\{0,1\}^{\mathbb{Z}}$: reading $01$ (resp. $10$) at position $i$ means that a replacement $0\to 1$ (resp $1\to 0$) occurs at position $i$.
Such a word $w$ is moreover unique, except if $u=v$ in which case both $w=0^{\mathbb{Z}}$ and $w=1^{\mathbb{Z}}$ suit.
This coding will be used in the proof of Prop.~\ref{prop:sofic_z'}.
Figures \ref{fig:close_words} and \ref{fig:close_sturmian_words} illustrate this.

\begin{figure}[hbtp]
\includegraphics[width=\textwidth]{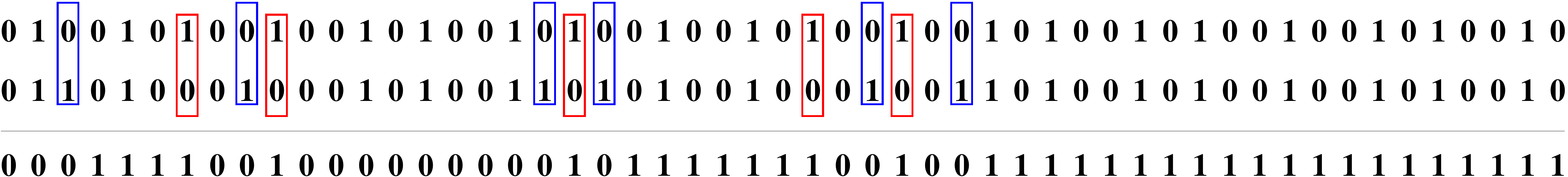}
\caption{Two (factors of) bi-infinite words with emphasized replacements (first two lines). The corresponding coding (last line).}
\label{fig:close_words}
\end{figure}

\begin{figure}[hbtp]
\includegraphics[width=\textwidth]{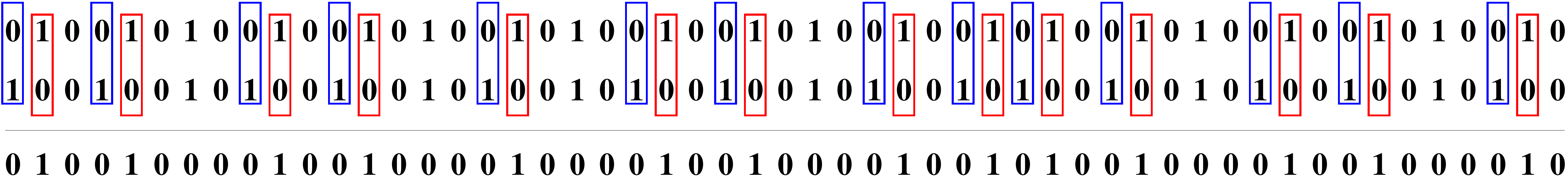}
\caption{Two (factors of) Sturmian words with the same slope and emphasized replacements (first two lines). The corresponding coding (last line).}
\label{fig:close_sturmian_words}
\end{figure}

\section{A Sofic subshift}
\label{sec:subshift}

The main result of \cite{AS,DRS} is phrased in terms of symbolic dynamics.
Let us first recall this formalism, which is convenient to keep, and then explain the correspondance with tilings. Given a finite alphabet $\mathcal{A}$, a {\em configuration} is a word indexed by $\mathbb{Z}^n$.
Consider a set of finite patterns $F$, the {\em subshift of forbidden patterns $F$} is the set of configurations where no pattern in $F$ appears.
A subshift $S$ is said to have {\em finite type} if there is a finite set of forbidden patterns.
A subshift $S$ is said to be {\em sofic} if there is a subshift of finite type $S'$ and a map from the alphabet of $S'$ to the alphabet of $S$, called {\em factor map}, which maps $S'$ onto $S$.
Last, a subshift is said to be {\em effective} if its forbidden patterns can be enumerated by a Turing machine.
In terms of tilings, a configuration can be seen as a tiling of $\mathbb{R}^n$ by tiles which are colored unit hypercubes.
Forbidden patterns then correspond to local rules, and the factor map simply correspond to a map on colors of tiles.
We are now in a position to use the result of \cite{AS,DRS}. It is (constructively) proven that any effective $d$-dimensional subshift can be obtained as the {\em projective subaction} of a $(d+1)$-dimensional sofic subshift, that is, the projection onto $d$ given coordinates.\\

First, let $Z_\alpha$ be the two-dimensional subshift whose configurations are obtained by copying on each row a given Sturmian word of slope $\alpha$:
$$
Z_{\alpha}=\left\{u\in\{0,1\}^{\mathbb{Z}^2},~\exists\rho\in\mathbb{R},~\forall m\in\mathbb{Z},~u(m,\cdotp)=s_{\alpha,\rho}\right\}.
$$
When $\alpha$ is computable, we compute the word $s_{\alpha,0}$ until to obtain $n+1$ distinct factors of size $n$, thus the finite patterns of $s_{\alpha,\rho}$ are recursively enumerated. It follows that the projective subaction of $Z_{\alpha}$ is effective, thus $Z_{\alpha}$ is sofic when $\alpha$ is computable by~\cite{AS,DRS}.
Say, {\em e.g.}, that $Z_{\alpha}$ is the projection onto the first coordinate of a finite type subshift $\tilde{Z}_\alpha$ of $(\{0,1\}\times B)$, for some finite alphabet $B$.\\

Then, extend $Z_\alpha$ by the subshift $Z'_\alpha$ whose elements have rows at distance at most one from the Sturmian word of slope $\alpha$ and intercept $0$:
$$
Z'_\alpha:=\left\{u\in\{0,1\}^{\mathbb{Z}^2},~\forall m\in\mathbb{Z},~d(u(m,\cdotp),s_{\alpha,0})\leq 1\right\}.
$$
Prop.~\ref{prop:parallel_sturmian} indeed yields $Z_\alpha\subset Z'_\alpha$ and that the choice of the intercept - here $0$ - has no importance.
Let us constructively prove:

\begin{proposition}\label{prop:sofic_z'}
The subshift $Z'_\alpha$ is sofic when $\alpha$ is computable.
\end{proposition}

\begin{proof}
Let $\pi_{i_1,\ldots,i_k}$ denotes the projection on the $i_1$-th,\ldots,$i_k$-th coordinates.
Let $\tilde{Z}'_\alpha$ be the subshift of $(\{0,1\}\times B\times \{0,1\})^{\mathbb{Z}^2}$ such that $u\in\tilde{Z}'_\alpha$ if and only if $\pi_{12}(u)\in \tilde{Z}_\alpha$ and, for any $(m,n)\in\mathbb{Z}^2$:
$$
\pi_3(u(m,n))<\pi_3(u(m,n+1)) ~\Rightarrow~ \pi_1(u(m,n))=0,
$$
$$
\pi_3(u(m,n))>\pi_3(u(m,n+1)) ~\Rightarrow~ \pi_1(u(m,n))=1.
$$
Clearly, $\tilde{Z}'_\alpha$ has finite type when so does $\tilde{Z}_\alpha$.
Now, we claim that the factor map $\pi$ defined as follows maps $\tilde{Z}'_\alpha$ onto $Z'_\alpha$:
$$
\pi(u)(m,n)=\left\{\begin{array}{cl}
\pi_1(u(m,n)) & \textrm{if }\pi_3(u(m,n))=\pi_3(u(m,n+1)),\\
1-\pi_1(u(m,n)) & \textrm{otherwise.}
\end{array}\right.
$$
Let $\tilde{u}\in\tilde{Z}'_\alpha$ and fix $m\in\mathbb{Z}$.
By definition of $\tilde{Z}'_\alpha$ and $\tilde{Z}_\alpha$, $\pi_1(\tilde{u}(m,\cdotp))=s_{\alpha,\rho}$.
One thus also has $\pi(\tilde{u}(m,\cdotp))=s_{\alpha,\rho}$, except at each position $n$ such that the two bits $\pi_3(\tilde{u}(m,n))$ and $\pi_3(\tilde{u}(m,n+1))$ differ.
At these positions, $\pi(\tilde{u}(m,\cdotp))$ is obtained by performing on $s_{\alpha,\rho}$ a replacement of type $\pi_3(\tilde{u}(m,n))\to \pi_3(\tilde{u}(m,n+1))$.
The type of these replacements alternate - as the bit runs do - and Prop.~\ref{prop:alternated_replacements} yields $d(\pi(\tilde{u}(m,\cdotp)),s_{\alpha,\rho})\leq 1$.
This shows that $\pi(\tilde{u})$ is in $Z'_\alpha$.
Hence, $\pi(\tilde{Z}'_\alpha)\subset Z'_\alpha$.\\

Conversely, let $u\in Z'_\alpha$. Fix $m\in\mathbb{Z}$ and choose $\tilde{v}\in \tilde{Z}_\alpha$ such that $\pi_1(\tilde{v}(m,\cdotp))=s_{\alpha,0}$. By definition, $d(u(m,\cdotp),s_{\alpha,0})\leq 1$, so we can consider $w_m$  the coding of the replacements which transform $s_{\alpha,0}$ into $u(m,\cdotp)$ (see end of Section~\ref{sec:quasisturmian}).
Consider $\tilde{u}\in (\{0,1\}\times B\times \{0,1\})^{\mathbb{Z}^2}$ defined by $\tilde{u}(m,i)=(\tilde{v}(m,i),w_m(i))$.
The way $\pi$ has been defined yields $\tilde{u}\in\tilde{Z}'_\alpha$ and $\pi(\tilde{u})=u$.
Hence, $Z'_\alpha\subset \pi(\tilde{Z}'_\alpha)$.
\end{proof}

Let us mention that the two-dimensionality plays a fundamental role in the result obtained in \cite{AS,DRS}, hence in the soficity of $Z_\alpha$.
It is thus also fundamental in the proof of Prop.~\ref{prop:sofic_z'}, although lines seems to be there only independantly considered.
The interplay between the lines of $\tilde{Z}'_\alpha$ is indeed ``hidden'' in the alphabet $B$ and in the forbidden patterns of $\tilde{Z}_\alpha$.

\section{Dimension two and codimension one}
\label{sec:3_vers_2}

Consider a computable $3\to 2$ planar tiling defined over vectors $\vec{v}_1$, $\vec{v}_2$ and $\vec{v}_3$.
Its slope can be defined by its normal vector, say $(1,\alpha,\beta)$, which is computable.\\
As mentioned above, any sofic subshift can be seen as tilings by a given tile set and local rules.
Here, it is convenient to model local rules by coloring edges of tiles and assuming that two tiles can be adjacent only on edges which have the same color (this is actually the definition used in \cite{wang}).
Let thus $\tau_\alpha$ and $\tau_\beta$ be such tile sets, with moreover label $0$ or $1$ on tiles, such that the tilings by $\tau_\alpha$ and $\tau_\beta$ respectively correspond, when considering only labels, to $Z'_\alpha$ and $Z'_\beta$.\\

\noindent Let us derive from $\tau_\alpha$ a tile set $\tau'_\alpha$ as follows (see Fig.~\ref{fig:square_to_rhombus}):
\begin{itemize}
\item each tile in $\tau_\alpha$ with label $0$ is sheared along $\vec{v}_3$ to give in $\tau'_\alpha$ a rhombus tile defined by $\vec{v}_2$ and $\vec{v}_3$ (edge colors are unmodified);
\item each tile in $\tau_\alpha$ with label $1$ is sheared along $\vec{v}_3$ to give in $\tau'_\alpha$ a rhombus tile defined by $\vec{v}_1$ and $\vec{v}_3$ (edge colors are unmodified);
\item a rhombus tile defined by $\vec{v}_1$ and $\vec{v}_2$ is in $\tau'_\alpha$ if and only if its edges have colors which appears on $\vec{v}_1$- or $\vec{v}_2$-edges of square tiles of $\tau_\alpha$, with the restriction that two edges meeting at $\vec{v}_1$ or $\vec{v}_2$ must have the same color ({\em transfer} tile).
\end{itemize}

\begin{figure}[hbtp]
\centering
\includegraphics[width=\textwidth]{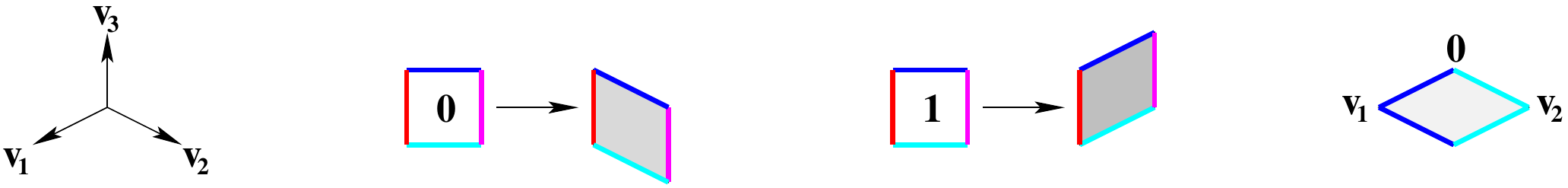}
\caption{The $\vec{v}_i$'s (left), the shearing of tiles in $\tau_\alpha$ (center) and the transfer tiles (right).}
\label{fig:square_to_rhombus}
\end{figure}

The idea behind the definition of $\tau'_\alpha$ is simple.
Let us call {\em $\vec{v}_i$-ribbon} of a $3\to 2$ tiling a maximal sequence of tiles, with two consecutive tiles being adjacent along an edge $\vec{v}_i$.
Then, one easily sees that $\tau'_\alpha$ forms the $3\to 2$ tilings whose $\vec{v}_3$-ribbons embed configurations in $Z'_\alpha$, with the transfer tiles just carrying colors between ribbons (see Fig.~\ref{fig:z_embedding}).

\begin{figure}[hbtp]
\includegraphics[width=\textwidth]{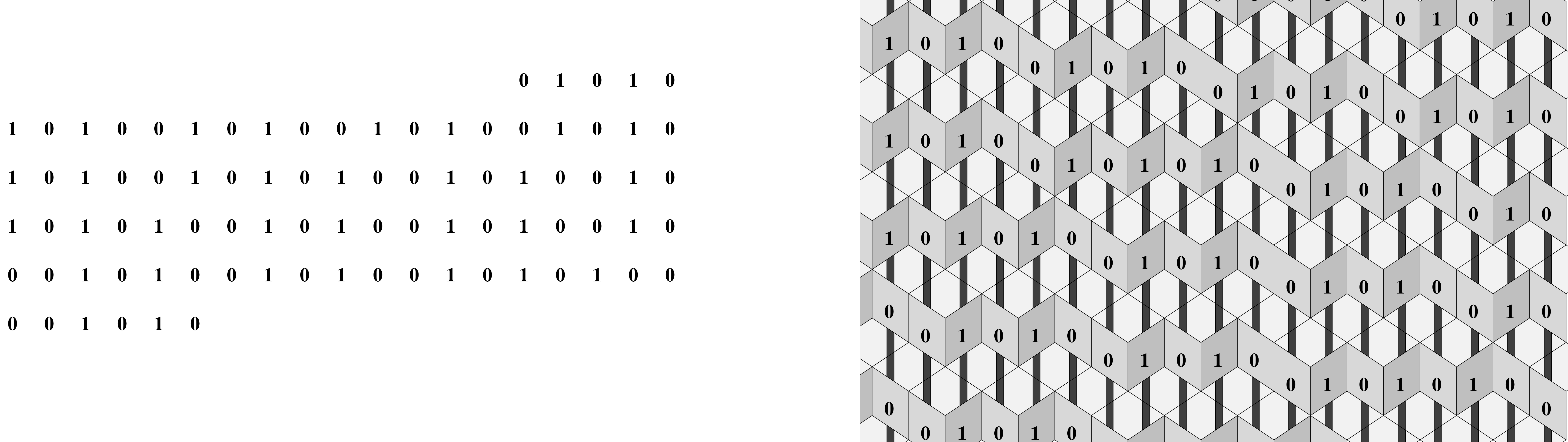}
\caption{A part of an element in $Z'_\alpha$ (left) and its embedding into the $\vec{v}_3$-ribbons of a $3\to 2$ tiling (right, with thick lines drawn on ``transfer'' tiles indicating the way tile decorations travel between neighbor $\vec{v}_3$-ribbons.)}
\label{fig:z_embedding}
\end{figure}

By proceeding similarly up to a circular permutation on the $\vec{v}_i$'s, one derives from $\tau_\beta$ a tile set $\tau'_\beta$ which forms the $3\to 2$ tilings whose $\vec{v}_2$-ribbons embed configurations in $Z'_\beta$.
Let us finally define the tile set $\tau'_{\alpha,\beta}$ as the cartesian product of $\tau'_\alpha$ and $\tau'_\beta$, that is, to each pair of identically shaped tiles in $\tau'_\alpha\times \tau'_\beta$ corresponds a tile in $\tau'_{\alpha,\beta}$ which has the same shape, with each edge having a color which encodes the colors of the pair of corresponding edges in $\tau'_\alpha$ and $\tau'_\beta$.
The tilings of $\tau'_{\alpha,\beta}$ are thus the $3\to 2$ tilings which embed $Z'_\alpha$ on their $\vec{v}_3$-ribbons and $Z'_\beta$ on their $\vec{v}_2$-ribbons.
This allows only planar $3\to 2$ tilings of slope $(1,\alpha,\beta)$ and width at most $4$ (this can be seen by decomposing any path between two points in a $\vec{v}_3$-ribbon followed by a $\vec{v}_2$-ribbon).
This moreover allows at least the planar $3\to 2$ tilings of slope $(1,\alpha,\beta)$ and width $1$.
We thus (constructively) proved:

\begin{proposition}\label{prop:3_vers_2}
Any computable planar $3\to 2$ tiling admits local rules.
\end{proposition}

\section{Higher dimensions and codimensions}
\label{sec:n_vers_d}

The last step to prove Th.~\ref{th:main} is to extend Prop.~\ref{prop:3_vers_2} to $n\to d$ tilings.
Although technical, this last step requires no new ideas.\\

For higher codimensions, we proceed by induction.
Our induction hypothesis is that any effective planar $n\to 2$ tiling admits weak local rules.
This holds for $n=3$ according to the previous section.
Let now $\mathcal{T}$ be an effective planar $(n+1)\to 2$ tiling.
For any basis vector $\vec{e}_i$, we project the lift of $\mathcal{T}$ along $\vec{e}_i$ to get the lift of an effective planar $n\to 2$ tiling, say $\mathcal{T}_i$.
By assumption, $\mathcal{T}_i$ admits local rules: let $\tau_i$ be a tile set whose tilings are at distance at most $w$ from $\mathcal{T}_i$.
We complete $\tau_i$ by adding the tiles with a $\vec{v}_i$-edge (that is, the tiles which disappeared from $\mathcal{T}$ by projecting along $\vec{e}_i$), with each of these tiles having no decoration on its $\vec{v}_i$-edges, and on the other edges a unique decoration that could be any of those appearing on an edge of a tile in $\tau_i$.
These new tiles thus just transfer decorations between the tiles of $\mathcal{T}_i$ (see Fig.~\ref{fig:higher_codim} for $n+1=4$).
Last, we define the tile set $\tau$ as the cartesian product of all the $\tau_i$'s (as we did for $\tau_\alpha$ and $\tau_\beta$ in the previous section).
This allows only $n+1\to 2$ tilings at distance at most $w'$ from $\mathcal{T}$ - in particular $\mathcal{T}$ itself.
This shows that $\mathcal{T}$ admits local rules.

\begin{figure}[hbtp]
\centering
\includegraphics[width=\textwidth]{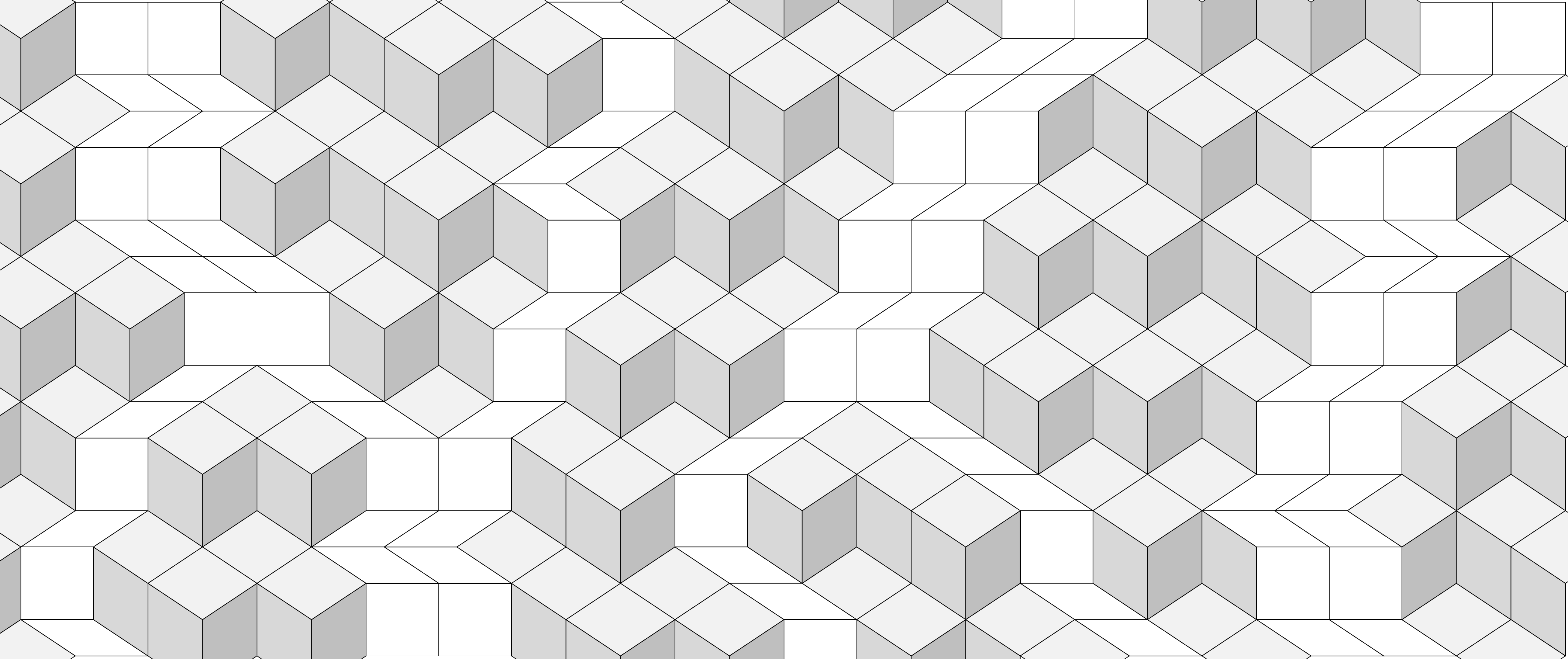}
\caption{
A $4\to 2$ tiling by shaded $\tau_i$ tiles and white additional tiles, with $\vec{v}_i$ being here horizontal.
The white tiles simply transfer horizontally the decorations of the shaded tiles.
By contracting each $\vec{v}_i$-edges to a point, the white ribbons disappear: we get a $3\to 2$ tiling by $\tau_i$, at distance at most $w'$ from $\mathcal{T}_i$.}
\label{fig:higher_codim}
\end{figure}

For higher dimensions, we also proceed by induction.
Our induction hypothesis is, for a fixed $n$, that any effective planar $n\to d$ tiling, $d<n$, admits local rules.
This holds for $d=2$ according to the above paragraph.
Let now $\mathcal{T}$ be an effective planar $n\to (d+1)$ tiling, with $d+1<n$.
Fix $i\in\{1,\ldots,n\}$.
For two tiles $T$ and $T'$ of $\mathcal{T}$, write $T\sim T'$ if these tiles share a $\vec{v}_i$-edge and let $\simeq$ be the transitive closure of the relation $\sim$.
Denote by $(\mathcal{T}_k)_{k\in\mathbb{Z}}$ the equivalence classes of $\simeq$, such that, for any $k$, $\mathcal{T}_k$ and $\mathcal{T}_{k+1}$ can be connected by a path which does not cross any other equivalence class (the $\mathcal{T}_k$'s play the role of $\vec{v}_i$-ribbons in the previous section).
By contracting all the $\vec{v}_i$-edges of a $\mathcal{T}_k$ ({\em flattening}), one gets a planar $n\to d$ tiling.
Its slope moreover depends only on the slope of $\mathcal{T}$, and in particular it is effective.
This allows to see $\mathcal{T}$ as a sequence of ``stacked'' parallel effective planar $n\to d$ tilings (namely the flattened $\mathcal{T}_k$'s), with the remaining tiles containing no $\vec{v}_i$-edge.
By induction, there exists a finite tile set $\tau_i$ whose tilings are at bounded distance $w$ from any of the flattened $\mathcal{T}_k$'s (since they are all parallel).
It is straighforward to ``unflatten'' $\tau_i$ to get a tile set $\tilde{\tau}_i$ whose tilings are at bounded distance $w$ from any of the $\mathcal{T}_k$'s.
We complete $\tilde{\tau}_i$ by adding the tiles without $\vec{v}_i$-edge (that is, the tiles lying between the stacked $\mathcal{T}_k$'s), with decorations being just transferred between consecutive $\mathcal{T}_k$'s along the direction $\vec{v}_i$ (as done in the previous section to transfer decorations between consecutive $\vec{v}_i$-ribbons).
The last step is (as in the previous section again) to define the cartesian product $\tau$ of the tile sets $\tilde{\tau}_i$, $i=1,\ldots,d$: its tilings are those at bounded distance $w$ from $\mathcal{T}$ - in particular $\mathcal{T}$ itself.
This shows that $\mathcal{T}$ admits local rules.

\end{document}